\documentclass[11pt,fleqn]{article}
\usepackage{amsmath,amssymb,graphicx,amsthm,mathtools}

\usepackage{fullpage}
\usepackage{multicol}

\usepackage[numbers]{natbib}

\usepackage[pagebackref]{hyperref}

\usepackage{xcolor}
\definecolor{dargray}{rgb}{0.18, 0.18, 0.18}
\definecolor{darkgreen}{rgb}{0.01,0.6,0.1}
\definecolor{lightrose}{rgb}{0.996,0.75,0.793}
\definecolor{rose}{cmyk}{0.75, 0.75, 0,0}
\definecolor{winered}{rgb}{0.6,0.1,0.1}
\definecolor{lightyellow}{rgb}{1, 1, 0.6}
\definecolor{transparent}{rgb}{1,1,1}
\definecolor{lightlightgray}{rgb}{0.88, 0.88, 0.88}
\definecolor{lightgray}{rgb}{0.8, 0.8, 0.8}
\definecolor{lightblue}{rgb}{0.527,0.805,0.977}
\definecolor{lightgreen}{rgb}{.74,1,0}

\usepackage{paralist}

\hypersetup{menucolor=orange!40!black,filecolor=winered,urlcolor=green!40!black, pdfencoding=auto, psdextra,
colorlinks=true,citecolor=green!60!black,linkcolor=winered,
pagebackref=true}

\newtheorem{theorem}{Theorem}
\newtheorem{corollary}{Corollary}

\newtheorem{example}{Example}
\newtheorem{lemma}{Lemma}

\newtheorem{proposition}{Proposition}
\newtheorem{claim}{Claim}

\newcommand{\ppp}{{\mathcal{P}}}

\newcommand{\MirkMin}{\probname{Mirkin Distance Minimization}}

\newcommand{\mirkin}{\textsf{mirk}}
\newcommand{\Mirkin}{\textsf{Mirk}}
\newcommand{\sol}{\ensuremath{s^*}}

\newcommand{\hamming}{\textsf{hd}}
\newcommand{\mpdist}{\ensuremath{\delta}}
\newcommand{\Hd}[1]{\textsf{HD}}
\newcommand{\ins}{\textsf{ins}}

\usepackage{needspace}

\usepackage[linecolor=green!70!black, backgroundcolor=green!10, bordercolor=black, textsize=scriptsize,textwidth=1.4cm]{todonotes}

\newcommand{\mirkValue}{L\cdot n\cdot \gvalue+\truthbound}
\newcommand{\bvalueS}{(4n-4)}

\newcommand{\gvalue}{\left[\binom{2n-2}{2}+(4n-4)\right] \cdot (2n-2)}
\newcommand{\truthbound}{m\cdot (3n^2-11)}
\newcommand{\LL}{3m\cdot n^2}

\newcommand{\prob}[6]{%
 \needspace{3\baselineskip}
  \begin{center}%
    \begin{minipage}{0.95\linewidth}%
      \textsc{#1}\\
      \textbf{#2} #3\\
      \textbf{#4} #5
    \end{minipage}%
  \end{center}%
}

\newcommand{\myemph}[1]{{\color{green!40!black}\emph{#1}}}
\newcounter{mycounter}

\usepackage{scrextend}
\usepackage{subeqnarray}

\newcommand{\probname}[1]{\textsc{#1}}

\newcommand{\probdef}[3]{\prob{\probname{#1}}{Input:}{#2}{Question:}{#3}{as}}

\newcommand{\ETH}{Exponential Time Hypothesis}

\newcommand{\pNAESAT}{\textsc{NAE-3SAT}}
\newcommand{\pHDClong}{\textsc{$p$-Norm Hamming Centroid}}

\usepackage[noend,ruled,linesnumbered]{algorithm2e} 
\usepackage[sort&compress,nameinlink,capitalize]{cleveref}
\crefname{algorithm}{Algorithm}{Algorithms}
\crefname{lemma}{Lemma}{Lemmas}
\crefname{proposition}{Proposition}{Propositions}
\crefname{theorem}{Theorem}{Theorem}
\crefname{claim}{Claim}{Claims}
\crefalias{AlgoLine}{line}

\title{A Note on Clustering Aggregation for Binary Clusterings}

\author{Jiehua Chen$^1$ and Danny Hermelin$^2$ and Manuel Sorge$^1$\\[2ex]
{\small $^1$TU Wien, Institute for Logic and Computation, TU Wien, Austria}\\
{\small $^2$Ben-Gurion University of the Negev, Beer Sheva, Israel}\\
{\footnotesize \texttt{jiehua.chen@ac.tuwien.ac.at} \quad \texttt{hermelin@bgu.ac.il}\quad \texttt{manuel.sorge@ac.tuwien.ac.at}}}

\begin{document}

\maketitle

\begin{abstract}
  We consider the clustering aggregation problem in which we are given
  a set of clusterings and want to find an aggregated clustering which
  minimizes the sum of mismatches to the input clusterings. In
  the binary case (each clustering is a bipartition) this problem was
  known to be NP-hard under Turing reductions. %
  We strengthen this result by providing a polynomial-time many-one
  reduction. Our result also implies that no
  $2^{o(n)}\cdot |I'|^{O(1)}$-time algorithm exists that solves any
  given clustering instance~$I'$ with $n$ elements, unless the \ETH{} fails.
  On the positive side, we show that the problem is fixed-parameter tractable with
  respect to the number of input clusterings and we give an integer linear programming formulation.
\end{abstract}

\section{Introduction}
Clustering is a fundamental data analysis task; in a basic form we aim to partition a set of given entities into groups of pairwise similar entities.
Nowadays, each entity is often specified by multiple attributes.
For example, in a social network a user may have a profile that consists of a description, a GPS or other location trace, and lists of friends, of liked posts, of visited websites, and so on. 
It is desirable to take aspects of these different attributes into account when computing a clustering.
This desire sparked the fields of so-called multi-view~\cite{yang_multiview_2018,p_multiview_2019,fu_overview_2020}, multi-layer~\cite{KL15,yuvaraj2021topological}, or ensemble-based clusterings~\cite{StrGho2002,FilSki2004,TAG17,boongoen_cluster_2018}.
A common strategy to obtain the overall clustering is then, to compute clusterings based on individual attributes and, afterwards, to \emph{aggregate} the clusterings into a single consensus clustering~\cite{StrGho2002,FilSki2004,GioManTsa2007,boongoen_cluster_2018}.
For instance, we may first individually cluster the social-network users according to their friendship relation graph, then cluster them based on meetings they had in the location trace, then on the similarity of visited websites, and finally aggregate all these clusterings into one.
The problem that we are interested in this article is the following: Given a set of clusterings, how to compute a consensus clustering?

A common interpretation of a \emph{consensus} clustering~$\mathcal{C}$ is one that is closest to the input clusterings by minimizing the sum over all input clusterings $\mathcal{D}$ of some distance measure between $\mathcal{C}$ and $\mathcal{D}$~(see the survey \cite{boongoen_cluster_2018}).
A fundamental way to measure the distance between two clusterings $\mathcal{C}$ and $\mathcal{D}$ is to count the number of pairs of entities that are clustered differently by $\mathcal{C}$ and $\mathcal{D}$~\cite{GioManTsa2007}.
That is, to count the number of pairs that are together in one cluster in $\mathcal{C}$ and in two different clusters in $\mathcal{D}$ and vice~versa.
This distance measure is known as the \emph{Mirkin distance}~\cite{Mirkin1996}.

We in particular focus on the case of binary clusterings, that is, clusterings into two clusters.
Thus, we arrive at the problem of, given a set $S$ of binary clusterings, to compute a binary clustering that minimizes the sum of Mirkin distances to the clusterings in $S$.
We call (the decision variant of) this problem \MirkMin.
See \cref{sec:prelim} for the formal definitions.

\paragraph{Our contributions}
Our main result in this paper is a tight running-time bound for the \MirkMin{} problem.
Specifically, there is a straightforward brute-force algorithm that solves the problem in $O(2^n \cdot nm)$ time, where $n$ denotes the length of the input strings and $m$ their number: The algorithm simply tries all $2^n$ solution strings and checks their Mirkin distance to the input.
  At first glance it seems a $O(n^2)$-time factor is needed to compute the Mirkin distance of a guessed solution string and the input strings.
  Note that for binary input strings, the Mirkin distance can be computed via the Hamming distance (see \cref{sec:prelim}), which can be computed in $O(n)$ time. %
Unfortunately, the running time of this brute-force algorithm cannot be substantially improved:
We show, using an intricate reduction, that the problem cannot be solved in $2^{o(n)} \cdot (nm)^{O(1)}$ time unless the Exponential Time Hypothesis (ETH)~\cite{ImPaZa2001} fails.
This settles the form of the asymptotic complexity with respect to the parameter~$n$.
In the second part of the paper, we consider the parameter $m$ of input strings, we show that the problem is fixed-parameter tractable with respect to~$m$, and we give an integer linear programming (ILP) formulation.

\paragraph{Related work}
\citet{DoGuKoWe2014} showed that \MirkMin{} is NP-hard under Turing reductions (see their Theorem~3).
That is, they gave a construction that takes an instance $I$ of the so-called \textsc{Cluster Editing} problem and produces a polynomial number of instances of \MirkMin{} such that one of these instances is positive if and only if $I$ is.
Instead, we show that we can map each instance of \textsc{NAE-3SAT} to a single equivalent one of \MirkMin. (\textsc{NAE-3SAT} is the variant of \textsc{3SAT} in which we want to a find an assignment of truth values to variables such that each clause contains a satisfied and an unsatisfied literal.)
Moreover, our prudent use of gadgets also allows us to give a tight lower bound for the running time assuming the ETH.

\MirkMin{} is a variant of the NP-hard \textsc{Clustering Aggregation}~\cite{GioManTsa2007} problem (aka.\ \textsc{Consensus Clustering}~\cite{FilSki2004} or \textsc{Clusters Ensembles}~\cite{StrGho2002}) from machine learning and bioinformatics.
Therein, we have as input a multiset~$\ppp$ of $m$~partitions on an $n$-element set~$U$ and the goal
is to search for a target partition~$P^*$ that minimizes the sum of Mirkin distances to all $m$ partitions.

Let us now view clusterings from a relational point of view.
Recall that a \myemph{partition} on the set~$U$ is an equivalence relation~$\sim$ (i.e., reflexive, symmetric, and transitive) over~$U\times U$.
Each partition can be represented by the equivalence classes of the corresponding equivalence relation.
We can thus alternatively define the Mirkin distance between two partitions as the number of pairs of elements which are equivalent in one partition but non-equivalent in the other, or the other way round.

\MirkMin{} has applications in voting theory and is also a restriction of the \textsc{Binary Relation Aggregation via Median Procedure} problem~\cite{BarMon1981,Wakabayashi1986,Wakabayashi1998,Hudry2012}.
The latter problem takes as input a set of alternatives~$C$ and a set of votes expressed as binary relations over~$C\times C$, and aims at finding a binary relation with minimum sum of symmetric-difference distances to the votes~\cite{BarMon1981}.
The \myemph{symmetric-difference distance} between two binary relations~$s$ and $t$ is simply the cardinality of their symmetric differences~$\mpdist(s,t)=|s\setminus t| + |t\setminus s|$.
It is straightforward to see that for equivalence relations the
symmetric-difference distance equals two times the Mirkin distance.

Relatedly, we considered a problem called \pHDClong, which, for some fixed $p>1$, is to find a centroid string which minimizes the $p$\nobreakdash-norm of its Hamming distances to the input strings~\cite{CheHerSor2018pnormesa}.
When the objective is to maximize instead of minimize the distances and when $p=2$, the \MirkMin{} problem can be reduced to this maximization variant.

\paragraph{Remarks}
In 2017, Jiehua Chen participated in the working group on ``Aggregation Procedures with Nonstandard Input and Output Types'' of the Dagstuhl seminar on ``Voting: Beyond Simple Majorities and Single-Winner Elections'' that our friend Gerhard Woeginger also attended~\cite{Dagstuhl2017}. One of the open questions to be addressed in the working group is concerned about the computational complexity of \MirkMin{} proposed by Bill Zwicker (Union College).
Gerhard immediately pointed out that the objective function underlying \MirkMin{} is not convex. Although the working group did not settle the computational complexity of \MirkMin{}, we would like to dedicate this article to Gerhard due to his insightful comments and discussions during the seminar. 

\section{Preliminaries} \label{sec:prelim}%
For an integer~$t$, let \myemph{$[t]$} denote the subset~$\{1,\ldots,t\}$.
For ease of presentation, we will throughout view the \MirkMin\ problem as a problem on binary strings, i.e., tuples in $\{0, 1\}^{*}$.
We abbreviate binary strings as \myemph{strings} if it is clear from the context.
Let $s$ and $s'$ be two strings.
Then, we use \myemph{$s\circ s'$} to denote the concatenation of $s$ and $s'$ 
and \myemph{$|s|$} and \myemph{$\overline{s}$} to denote the length and the complement of~$s$.
By \myemph{$s[i]$} we mean the value of string~$s$ at the $i^{\text{th}}$ coordinate.
and we write \myemph{$s[i,j]$} as shorthand of $s[i]s[j]$.
Given two integers~$i,j\in \{1,2,\ldots, |s|\}$ with $i \le j$, 
we use the notation \myemph{$s|^{j}_{i}$} to denote the substring~$s[i]s[i+1]\cdots s[j]$.
Further, let \myemph{$\hamming(s,s')$} denote the Hamming distance between strings~$s$ and $s'$,
i.e., the number of coordinates at which the values of $s$ and $s'$ differ.
For instance, $\hamming(0101, 1100) = 2$.

Given two binary strings~$s$ and $s'$, and an integer~$i$ with $1\le i \le |s|+1$, 
by $\ins(s, s', i)$ we mean the string obtained by inserting the string~$s'$ into $s$ just before the $i^{\text{th}}$~position.
For instance, $\ins(0110,00,3)=010010$.
In particular, $\ins(s,s',1)=s'\circ s$ and $\ins(s,s',|s|+1)=s\circ s'$.

The Mirkin distance~$\mirkin(s, s')$~\cite{Mirkin1996} between two equal-length strings~$s$ and~$s'$ counts the number of mismatches for each pair of coordinates. Formally, 
\myemph{${\Mirkin(s,s')}$} $=
\{\{i,j\}\colon (s[i]=s[j] \wedge s'[i]\neq s'[j])
\vee (s[i] \neq s[j] \wedge s'[i] = s'[j])\}$, and
$\myemph{\mirkin(s, s')}=|\Mirkin(s,s')|$.
For instance, 
$\Mirkin(01\,001, 01\,110) = \{\{i,j\} \mid i\in \{1,2\}\wedge j\in \{3,4,5\}\}$, so $\mirkin(01\,001, 01\,110) = 6$.
For binary strings, the Mirkin distance has an alternative definition that uses Hamming distances:
\[\mirkin(s,s')=\hamming(s,s')\cdot \hamming(\overline{s},s') = \hamming(s,s')\cdot (n-\hamming{}(s,s')).\]
In this paper, we will use the alternative definition extensively.
Note that, by the above formulation, the Mirkin distance function is not convex.

Let $S$ be a collection of strings of length~$n$ each and let $\sol$ denote an arbitrary length-$n$ string. 
The Mirkin distance between~$\sol$ and~$S$ is the sum of the Mirkin distances between $\sol$ and each string in the sequence: \myemph{$\mirkin(\sol,S)$} $=\sum_{s'\in S}\mirkin(\sol,s')$. %
The Mirkin distance between $\sol$ and $S$ regarding a pair~$\{i,j\}$, $i\neq j\in [n]$,
is the Mirkin distance between $\sol[i,j]$ and the multiset~$\{s'[i,j] \mid s'\in S\}$.

The formal statement of the problem is as follows:
\probdef{\MirkMin}
{A collection~$S$ of strings~$s_1, \ldots, s_m \in \{0, 1\}^n$ and an integer~$k$.}
{Is there a string $\sol \in \{0, 1\}^n$ such that $\mirkin(\sol, S) \le k$?}

\section{NP-hardness for \MirkMin}

In this section, we show that \MirkMin{} is indeed NP-hard by utilizing a gadget that \citet{DoGuKoWe2014} used to enforce that the Mirkin distance for each two coordinates, when restricted to only these two coordinates, is exactly half the number of the input strings since exactly half of the input strings have the same value ($00$ or $11$) and the other half have different values~($01$ or $10$). \cref{alg:build-gadget} computes such a kind of gadget. %
Note that, however, this type of gadget alone is not enough to devise a many-one hardness reduction.  %
This gadget can be used to encode truth-values of variables in a reduction from \textsc{NAE-3SAT} but an essential difficulty that remains is to find gadgets that encode clause satisfaction. 

\begin{algorithm}
  \DontPrintSemicolon
  \caption{Algorithm for constructing $2^{\ell}$  binary strings of length~$2^{\ell}$ each such that for each two coordinates,
    half of the strings have the same value and the other half not.}
  \label{alg:build-gadget}
  \small
  \SetKwBlock{Block}
  \SetAlCapFnt{\footnotesize}
  \SetKwFunction{RecursiveBuild}{Build}
  \RecursiveBuild{$2^{\ell}$}:
  \Block{
    \lIf{$\ell=1$}
    {
      \Return{$(00,01)$}
    }
    \Else{
      $(s_1,s_2,\ldots, s_{2^{\ell-1}})\leftarrow \RecursiveBuild(2^{\ell-1})$

      \Return{$(s_1\circ s_1, s_1\circ \overline{s}_1, \, \ldots, \, s_{2^{2\ell-1}}\circ s_{2^{2\ell-1}}, s_{2^{2\ell-1}} \circ \overline{s}_{2^{2\ell-1}} )$}
    }
  }
\end{algorithm}

We show that the strings constructed by \cref{alg:build-gadget} fulfills our requirement above.

\begin{proposition}\label{prop:alg-half}
  Let $S$ be the sequence of strings constructed by \cref{alg:build-gadget}.
  Then, for each two distinct coordinates~$i,j\in \{1,2,\ldots,2^{\ell}\}$,
  \begin{enumerate}[(1)]
    \item there are $|S|/2$ strings from $S$, called~$a_{1}$, $a_2$, $\ldots, a_{|S|/2}$, such that $a_r[i]=a_r[j]$, $r\in [|S|/2]$, and
    \item there are $|S|/2$ strings from $S$, called $b_1$, $b_2$, $\ldots, b_{|S|/2}$, such that $b_r[i]\neq b_r[j]$, $r\in [|S|/2]$.
\end{enumerate} 
\end{proposition}

\begin{proof}
  We show the statement via induction on $\ell$.
  For $\ell=1$, \cref{alg:build-gadget} returns $(00,01)$. Our two statements follow immediately.
  Assume that sequence~$S'=\RecursiveBuild(2^{\ell-1})$ satisfies the proposition and let $S'=(s_1,\ldots, s_{2^{\ell}-1})$.
  We show that $S=\RecursiveBuild(2^{\ell})$ also satisfies the proposition.
  By \cref{alg:build-gadget}, we have $S=(s_r \circ s_r, s_r \circ \overline{s}_r)_{s_r \in S'}$.
  
  Consider two distinct coordinates~$i,j\in \{1,2,\ldots, 2^\ell\}$.
  Obviously, by our induction assumption, the two statements hold if $1\le i,j \le 2^{\ell-1}$ or $2^{\ell-1}+1\le i, j \le 2^\ell$.
  Thus, we assume that $1\le i \le 2^{\ell-1}$ and $2^{\ell-1}+1\le j \le 2^\ell$ (the other case when $1\le j \le 2^{\ell-1}$ and $2^{\ell-1}+1\le i \le 2^\ell$ is analogous).
  By construction, $S$ consists of the following strings $s_r\circ s_r$ and $s_r\circ \overline{s}_r$, $r \in [2^{\ell-1}]$.
  By our choice of $i$ and $j$ it follows that both strings have the same value at coordinate~$i$ and a different value at coordinate~$j$. 
  Hence, for each~$r\in [2^{\ell-1}]$, one of the strings from $\{s_r \circ s_r, s_r\circ \overline{s}_r\}$ has the same value at~$i$ and
  $j$, and the other does not.
  The two statements follow immediately.  
\end{proof}

To show NP-hardness, we reduce from \textsc{Not-All-Equal 3-SAT (NAE-3SAT)}~\cite{GJ79}, which, given a set of size-three clauses, asks whether there is a \emph{satisfying} truth assignment, that is, each clause has at least one true literal and at least one false literal.

\begin{theorem}\label{thm:mirkin-hard}
  \MirkMin{} is NP-hard.
\end{theorem}

\begin{proof}
  As mentioned, we reduce from the NP-hard NAE-3SAT problem~\cite{GJ79}.
  Let $I=(X,\mathcal{C})$ be an instance of NAE-3SAT, where $X=\{x_1,\ldots,x_n\}$ denotes a set of $n$~variables and $\mathcal{C}=\{c_1,\ldots,c_m\}$ denotes a set of $m$ clauses of size three each.
  By introducing variables that do not occur in any clauses, we assume without loss of generality that $n={\color{red}2^{\ell}+1}$ for some~$\ell$.
  We construct two groups of binary strings where each string is of length~$2n=2^{\ell+1}+2$. %
  Each variable will be encoded by a pair of two consecutive coordinates in the string, one at an odd position, one at an even position.
  We use the gadget constructed via \cref{alg:build-gadget} to enforce that these two coordinates will always have the same value so that $11$ will correspond to setting the variable to true while $00$ will correspond to setting the variable to~false.

  The strings are built on two groups of strings.
  \begin{description}
    \item[Group 1.]     
    Let $S=\RecursiveBuild(2^{\ell+1})$ (see \cref{alg:build-gadget}); note that $|S|=2n-2$.
    Then, for each integer~$r\in [n]$ (representing the index of a specific variable)
    we introduce $2^{\ell+1}$ strings as follows. 
    For each string~$s_i\in S$, construct two strings with the forms~$\ins(s_i, 11, 2r-1)$ and $\ins(\overline{s}_i, 11, 2r-1)$.
    Note that each of the newly constructed strings has length ${\color{red}2^{\ell + 1}+2} = 2n$.
    Let $S_r$ denote the sequence that contains the newly introduced strings.

    For instance, for $r=2$, $\ell=1$,
    sequences~$S$ and $S_r$ consist of $2^{\ell+1}=4$ and $2\cdot 2^{\ell+1}=8$ strings, respectively:
    \begin{multicols}{2}
      \noindent
      \begin{align*}
     S\colon & 00\,00,\\
      &  00\,11,\\
      &  01\,01,\\
      &  01\,10.
    \end{align*}    
    \columnbreak
      \begin{align*}
      S_r\colon &  00\,{\color{red}11}\,00,\\
      &  00\,{\color{red}11}\,11,\\
      &  01\,{\color{red}11}\,01,\\
      &  01\,{\color{red}11}\,10,\\[1ex]
      &  11\,{\color{red}11}\,11,\\
      &  11\,{\color{red}11}\,00,\\
      &  10\,{\color{red}11}\,10,\\
      &  10\,{\color{red}11}\,01.
    \end{align*}    
  \end{multicols}
  \item[Group 2.]
    For each clause~$c_j\in \mathcal{C}$, let $\ell_1,\ell_2,\ell_3$ be the three literals contained in $c_j$.
    We define three strings~$t^{(1)}_j, t^{(2)}_j, t^{(3)}_j \in \{0,1\}^{2 n}$ as follows:
    \begin{align*}
    \forall (i,z) \in [n]\times [3]\colon &\\
      t^{(z)}[2i-1,2i] = &
      \begin{cases}
        11, & \ell_z = x_i,\\
        00, & \ell_z = \overline{x}_i,\\
        00, & \ell_{y} = x_i \text{ for some } y \in [3]\setminus \{z\},\\
        11, & \ell_{y} = \overline{x}_i \text{ for some } y \in [3] \setminus \{z\},\\
        01, & \text{otherwise.}
      \end{cases}
    \end{align*}
    \noindent For instance, for clause~$c_j = (\overline{x}_1, x_2, \overline{x}_3)$, the three corresponding strings are
    \begin{align*}
      t^{(1)}_j = {\color{red}00}\,00\,11\,01\,01\,\ldots\,01,\\
      t^{(2)}_j = 11\,{\color{red}11}\,11\,01\,01\,\ldots\,01,\\
      t^{(3)}_j = 11\,00\,{\color{red}00}\,01\,01\,\ldots\,01.
    \end{align*}
    Let $T_j=\{t^{(1)}_j, t^{(2)}_j, t^{(3)}_j\}$.
  \end{description}
  Let $L=\LL$. Then, let $I'$ be an instance consisting of the following strings:
  For each $r\in [n]$, add $L$ copies of $S_r$ to $I'$.
  For each $j\in [m]$, add $T_j$ to $I'$.
  To complete the construction, let $k=\mirkValue$.
  Clearly, the construction can be done in polynomial time since $\RecursiveBuild(2^{\ell + 1})$ takes $O(2^{2\ell + 2})=O(n^2)$ time.
  
  We claim that the instance~$I$ has a satisfying truth assignment, that is, each clause has a true literal and a false literal, if and only if there is a binary string~$s$ that has a Mirkin distance of at most~$k$ to the strings from $I'$.

  Before we show the correctness of the construction, we present two observations which will help us to determine the solution string for $I'$.

  \begin{claim}\label{claim:sol-00or11}
    Let $s^*$ be an arbitrary binary string of length~$2n$.
    For each integer~$r \in [n]$, the following holds.
    \begin{enumerate}[(1)]
    \item If $s^*[2r-1,2r]\in \{01,10\}$, 
      then $\mirkin(s^*,S_r) = \gvalue + \bvalueS$.
    \item If $s^*[2r-1,2r]\in \{00,11\}$, 
      then $\mirkin(s^*,S_r) = \gvalue$.
    \end{enumerate}
  \end{claim}
  \begin{proof}
    \renewcommand{\qedsymbol}{(of \cref{claim:sol-00or11})~$\diamond$}
    By the construction of $S_r$, we have the following (see also \cref{prop:alg-half}).
    \begin{itemize}[--]
      \item For each pair~$\{i,j\}\subseteq [2n]\setminus \{2r-1,2r\}$ we have
    \begin{compactenum}
      \item $|S_r|/2$ strings~$s$ from $S_r$ such that $s[i]=s[j]$, and 
      \item $|S_r|/2$ strings~$s$ from $S_r$ such that $s[i]\neq s[j]$.
    \end{compactenum}
    This means that the Mirkin distance from $s^*$ to $S_r$ regarding the 
    pair~$\{i,j\}$ is always $|S_r|/2$.
    
    \item For each coordinate~$i \in [2n]\setminus \{2r-1,2r\}$, it holds that $|S_r|/2$ strings from $S_r$ have a $0$ in column~$i$ and $|S_r|/2$ strings from $S_r$ have a $1$ in column~$i$.
    Since all strings have a~$1$ in column~$2r-1$ (resp.\ $2r$), the Mirkin distance from $s^*$ to $S_r$ regarding the pair~$\{i, 2r-1\}$ (resp.\ $\{i,2r\}$) is also $|S_r|/2$.

    \item The Mirkin distance from $s^*$ to $S_r$ regarding the pair~$\{2r-1,2r\}$ is $|S_r|$ if $s^*[2r-1,2r]\in \{01,10\}$; otherwise it is zero.

  \end{itemize}
    In total, if $s^*[2r-1,2r]\in \{01,10\}$, then we have 
    \begin{align*}
      \mirkin(s^*, S_r)
      & = \binom{2n-2}{2}\cdot \frac{|S_r|}{2}+ 2\cdot (2n-2) \cdot \frac{|S_r|}{2} + |S_r| \\
      & = \gvalue+\bvalueS;
    \end{align*}
   otherwise, we have  
    \begin{align*}
      \mirkin(s^*, S_r)
      & = \binom{2n-2}{2}\cdot \frac{|S_r|}{2}+2 \cdot (2n-2) \cdot \frac{|S_r|}{2}\\
      & = \gvalue.
    \end{align*}
  \end{proof}
  
  \noindent
  Below we will sometimes interpret binary strings of length~$n$ as truth assignments to the variables in $X$, where the $i^{\text{th}}$ character assigns the corresponding truth value to variable $x_i$.
  Define also $\gamma\colon \{0,1\}^n \to \{0,1\}^{2n}$ by $\gamma(e_1\cdots e_n) = (e_1 e_1 \cdots e_n e_n)$.
  
  \begin{claim}\label{claim:truth}
    Let $c_j \in \mathcal{C}$ be an arbitrary clause.
    Then, for each $s\in \{0,1\}^n$, we have that $\mirkin(\gamma(s), T_j)\ge 3 n^2 - 11$.
    Moreover, the equality is attained if and only if the string~$s$, interpreted as a truth assignment to the variables $x_i, i \in \{1, \ldots, n\}$, satisfies $c_j$ with at least one true literal and at least one false literal.
  \end{claim}
  \begin{proof}\renewcommand{\qedsymbol}{(of \cref{claim:truth})~$\diamond$}
    Assume, without loss of generality, that the literals in $c_j$ correspond to the first, the second, and the third variable $x_1$, $x_2$, and $x_3$ (each in either a positive or a negative form).
    For each string~$t^{(z)}_j \in T_j$ with $z\in \{1,2,3\}$, 
    by the definition of the Hamming distance, 
    $\hamming(\gamma(s), t^{(z)}_j)=\hamming(\gamma(s)|^6_1, t^{(z)}_j|^6_1) + \hamming(\gamma(s)|^{2n}_{7}, t^{(z)}_j|^{2n}_{7})$.
    By the definition of $t^{(z)}_j$ regarding the positions from $7$ to $2n$, we have that
     $\hamming(\gamma(s)|^{2n}_{7}, t^{(z)}_j|^{2n}_{7})=n-3$.
     
     Assume that $s$ satisfies $c_j$ with the $a^{\text{th}}$ literal being true and the $b^{\text{th}}$ literal being false, $a, b \in \{1,2,3\}$ and $a\neq b$. 
     Let $z\in \{1,2,3\}\setminus \{a,b\}$.
     We distinguish between two cases.
     If $\ell_z$ is true under $s$ (i.e., the $z^{\text{th}}$ coordinate of string~$s$ is a $1$ if $\ell_z=x_z$ and a $0$ otherwise), then $\hamming{}(\gamma(s)|^6_1, t^{(a)}_j|^6_1)=2=\hamming{}(\gamma(s)|^6_1, t^{(z)}_j|^6_1)$ while $\hamming{}(\gamma(s)|^6_1, t^{(b)}_j|^6_1)=6$.
     If $\ell_z$ is false under $s$ (i.e., the $z^{\text{th}}$ coordinate of string~$s$ is a $1$ if $\ell_z=\overline{x}_z$ and a $0$ otherwise),
     then $\hamming{}(\gamma(s)|^6_1, t^{(a)}_j|^6_1)=0$  while $\hamming{}(\gamma(s)|^6_1, t^{(b)}_j|^6_1)=\hamming{}(\gamma(s)|^6_1, t^{(z)}_j|^6_1) = 4$.
     In other words, we have that $\{\hamming{}(\gamma(s), t^{(a)}_j),\hamming{}(\gamma(s), t^{(b)}_j)\} = \{n-1, n+3\}$ if $\ell_z$ is true under $s$; otherwise $\{\hamming{}(\gamma(s), t^{(a)}_j),\hamming{}(\gamma(s), t^{(b)}_j)\} = \{n-3, n+1\}$.
     In both cases, $\hamming(\gamma(s), t^{(z)}_j)\in \{n-1, n+1\}$. 
     In both cases, using the alternative definition of the Mirkin distance,
     we thus have that     
     \begin{align*}
       \mirkin(\gamma(s),T_j) & = (n^2-9)+(n^2-1)+(n^2-1) = 3n^2 - 11. 
     \end{align*}

     Now assume that under $s$ either all literals from $c_j$ are true or all literals from $c_j$ are false.
     For the first case, for each $z\in \{1,2,3\}$, we have $\hamming{}(\gamma(s)|^6_1, t^{(z)}_j|^6_1)=4$, implying $\mirkin(\gamma(s), t^{(z)}_j) = \hamming{}(\gamma(s), t^{(z)}_j)\cdot (2n-\hamming{}(\gamma(s), t^{(z)}_j)) = n^2-1$.
     For the other case,  for each $z\in \{1,2,3\}$, we have $\hamming{}(\gamma(s)|^6_1, t^{(z)}_j|^6_1)=2$, implying $\mirkin(\gamma(s), t^{(z)}_j) = \hamming{}(\gamma(s), t^{(z)}_j)\cdot (2n-\hamming{}(\gamma(s), t^{(z)}_j)) = n^2-1$ as well.
     Altogether, we have  $\mirkin(\gamma(s), T_j) =3(n^2-1) > 3n^2-11$.
 \end{proof}

 Now, we are ready to show the correctness, that is,
 $I=(X,\mathcal{C})$ admits a truth assignment such that each clause in $\mathcal{C}$ has a true literal and a false literal if and only if there is a string~$s^*$ whose Mirkin distance to the strings in $I'$ is at most $k = \mirkValue$.

 For the ``only if'' direction, assume that $s\in \{0,1\}^{n}$ is a satisfying assignment for $\mathcal{C}$ such that each clause~$c_j \in \mathcal{C}$ has at least one true literal and at least one false literal.
 \cref{claim:truth} indicates that $\gamma(s)$ has Mirkin distance $3n^2-11$ to each triple in $T_j$ that corresponds to a clause~$c_j$.
 The second statement in \cref{claim:sol-00or11} indicates that $\gamma(s)$ has Mirkin distance $\gvalue$ to all strings in $S_r$ that correspond to the variable~$x_r$. 
 Altogether, the Mirkin distance between $\gamma(s)$ and all strings in $I'$ is
 $m\cdot (3 n^2-11) + L\cdot n \cdot \gvalue = k$, as desired.

 For the ``if'' direction, assume that $s^*\in \{0,1\}^{2n}$ is a string whose Mirkin distance to all strings in $I'$ is at most~$k$.
 First, we claim that $s^*$ has the form~$s^*=e_1e_1\cdots e_ne_n$ with $e_i \in \{0,1\}$ for all $1\le i \le n$.
 Suppose, towards a contradiction, that $s^*$ is not of the desired form, and let $i\in [n]$ be an integer such that $s^*[2i-1,2i]\in \{01,10\}$.
 Then, by the first statement in \cref{claim:sol-00or11}, 
 the Mirkin distance of $s^*$ to the first group of strings in $I'$ will be at least
 \begin{align*}
  L \cdot n \cdot \gvalue + L\cdot \bvalueS, 
\end{align*}
which exceeds the distance bound~$k=\mirkValue$ since $L>\truthbound$ and $n \geq 2$, a contradiction.

Thus, $s^*$ has the form~$s^*=e_1e_1\cdots e_ne_n$ with $e_i \in \{0,1\}$ for all $1\le i \le n$.
We show that $s=e_1\cdots e_n$ is a satisfying assignment for $\mathcal{C}$ such that each clause has at least one true literal and at least one false literal.
By the above reasoning, the Mirkin distance of $s^*$ to the second group of strings can be at most $\truthbound$.
Since there are $m$~triples in the second group, one for each clause, the average Mirkin distance of $s^*=\gamma(s)$ to each triple is $3n^2-11$.
By \cref{claim:truth} the Mirkin distance of $s^*$ to each triple in the second group is indeed $3n^2-11$, meaning that under~$s$ each clause has at least one true literal and one false literal, as desired.
\end{proof}

The running time lower bound for \MirkMin{} relies on the following proposition.

\begin{proposition}\label{prop:ETH-NAE}
  Unless the \ETH{} fails, \pNAESAT\ does not admit any sub-exponential time~$2^{o(n+m)}\cdot (n+m)^{O(1)}$ algorithm, where $n$ and $m$ denote the number of variables and clauses respectively.
\end{proposition}
\cref{prop:ETH-NAE} follows from Theorem~46 Point~7 by \citet{JonssonLNZ17} where the authors in particular show that, if the satisfiability problem for a finite constraint language is NP-hard, then a subexponential-time algorithm for this satisfiability problem would refute the \ETH{}, even in the case where each variable occurs only a constant number of clauses.

As a corollary, we obtain a running time lower bound for our problem.
\begin{corollary}
  Unless the \ETH{} fails, there is no algorithm that solves any given instance~$I'$ of \MirkMin{} in time $2^{o(\hat{n})}\cdot |I'|^{O(1)}$ where $\hat{n}$ is the length of the input strings.
\end{corollary}

\begin{proof}
  To show the statement, note that the length~$\hat{n}$ of the strings that we constructed
  in the proof of \cref{thm:mirkin-hard} is exactly $2n$, where $n$ is the number of variables in the \pNAESAT\ instance.
  Hence, a sub-exponential running time for our problem will contradict \cref{prop:ETH-NAE}.
\end{proof}

\section{Parameter Number \boldmath$m$ of Strings and an Integer Linear Programming Formulation}

In this section, we show that \MirkMin\ is fixed-parameter tractable with respect to the number~$m$ of input strings and we give an integer linear programming (ILP) formulation.
To achieve this, we show that a solution can be represented by a binary string whose entries correspond to the column types of the input (to be defined shortly).
Interpreting this string as a set of variables and trying all assignments of values it will then follow that \MirkMin\ is solvable in time~$O(2^{2^m}\cdot m \cdot n)$.

The ILP also builds on the above-mentioned set of variables.
We note that an integer programming approach similar to ours is applicable in many string problems whenever the columns of the input can be grouped together in order to be represented by a constant number of variables~\cite{GraNieRoss2003,CheHerSor2018pnormesa}.
Here, however, the resulting mathematical programming formulation is not linear at first because the straightforward way to model the Mirkin distance involves multiplications of binary variables.
We give additional reformulation tricks such that we can safely omit the square of binary variables, and such that we can introduce some extra variables to avoid multiplications of binary variables, resulting in an integer linear programming formulation.

Before presenting the ILP formulation, we observe a useful property of an optimal solution that allows us to introduce only binary variables, one for each column type.
Herein, given a non-empty sequence~$S=(s_1,\ldots,s_m)$ of length\nobreakdash-$n$ strings, we say that two columns~$j,j'\in [n]$ have the \myemph{same type} if for each $i \in [m]$ it holds that $s_i[j]=s_i[j']$.
The \emph{type} of column $j$ is its equivalence class in the same-type relation.
Thus, each type is represented by a vector in $\{0, 1\}^m$.

\begin{lemma}\label{lem:same-type-same-value}
  Let $S$ be a sequence of $m$ strings, each of length~$n$,
  and let $\sol$ be a solution with minimum Mirkin distance to~$S$.
  If two distinct columns~$j$ and $j'$ with $j,j'\in [n]$ have the same type with respect to~$S$,
  then it holds that $\sol[j]=\sol[j']$.
\end{lemma}

\begin{proof}
  Towards a contradiction, suppose that $\sol[j]\neq \sol[j']$.
  We will show that making these two columns have the same values, either $0$ or $1$, will result in a better solution, i.e., a string with smaller Mirkin distance. 
  Let $\sol_{00}$ (resp.\ $\sol_{11}$) be a string that we obtain from $\sol$ by replacing with $00$ (resp.\ $11$) the values at positions~$j$ and~$j'$.
  Formally, we have $\sol_{00}[j,j']=00$ and $\sol_{11}[j,j']=11$, and for each $\ell\in [n]\setminus \{j,j'\}$, we have $\sol_{00}[\ell]=\sol_{11}[\ell]=\sol[\ell]$.
  Given two strings~$s$ and $t$, we define a function~$f$ that computes the Mirkin distance from $s$ to $S$ subtracted by the Mirkin distance from~$t$ to $S$:
  \begin{align*}
    f(s,t,S) \coloneqq  \mirkin(s, S) - \mirkin(t, S).
  \end{align*}
  To obtain a contradiction, we show that $\sol$ is not an optimal solution by showing that
  \begin{align*}
    f(\sol, \sol_{00}, S) + f(\sol, \sol_{11}, S) > 0,
  \end{align*}
  because this implies that
  \[\mirkin(\sol, S) > \mirkin(\sol_{00}, S) \text{ or }
    \mirkin(\sol, S) > \mirkin(\sol_{11}, S).\]

  For each input string~$s_i\in S$, let $d_i$ denote the Hamming distance between $\sol$ and $s_i$, restricted to the columns that are \emph{neither} $j$ nor $j'$.
  We show that $f(\sol, \sol_{00}, S)+f(\sol, \sol_{11})> 0$.
  \begin{alignat*}{3}
  &  f(\sol, \sol_{00}, S) + f(\sol, \sol_{11}, S) \nonumber\\
   &  = 2\mirkin(\sol, S)  - \mirkin(\sol_{00}, S) - \mirkin(\sol_{11}, S) \nonumber \\
     & =
     2 \sum_{\mathclap{s_i\in S}}(d_i+1)(n-d_i-1)\\
     & \quad -   \sum_{\mathclap{\substack{s_i\in S\\s_i[j,j']=00}}}d_i(n-d_i) - \sum_{\mathclap{\substack{s_i\in S\\s_i[j,j']=11}}}(d_i+2)(n-d_i-2)\nonumber\\
    & \quad     -   \sum_{\mathclap{\substack{s_i\in S\\s_i[j,j']=00}}}(d_i+2)(n-d_i-2) - \sum_{\mathclap{\substack{s_i\in S\\s_i[j,j']=11}}}d_i(n-d_i)\nonumber\\
    & = \sum_{\mathclap{s_i\in S}}\biggl(2(d_i+1)(n-d_i-1) - d_i(n-d_i) - (d_i+2)(n-d_i-2) \biggr)\\
    & = 2m > 0.\nonumber
  \end{alignat*}
 By our reasoning before, this implies that $\sol$ is not an optimal solution, a contradiction.
\end{proof}

By \cref{lem:same-type-same-value}, for each type of column, we only need to store whether the output string will contain 0 or 1 at each column that corresponds to this type.
Let $n'$ denote the number of different (column) types in~$S$.
Then, $n'  \le  \min(2^m, n)$.
Enumerate the $n'$ column types as $t_1, \ldots, t_{n'}$.
Below we identify a column type with its index for easier notation. Using this, we can encode the set~$S$ succinctly by introducing a constant~\myemph{$e[j]$} for each column type~$j \in [n']$ that denotes the number of columns with type~$j$.
Analogously, given an optimal solution string~$\sol$, by \cref{lem:same-type-same-value} we can also encode this string~$\sol$ via a binary vector~$x\in \{0,1\}^{n'}$, 
where for each column type~$j \in [n']$ we use $x[j]$ to indicate whether the columns that correspond to the type have zeros or ones. Note that this encodes all essential information in a solution, since the actual order of the columns is not important.

\begin{example}
For an illustration, let $S=\{0000,0001,1110\}$.
Set~$S$ has two different column types, represented by $(0, 0, 1)^T$%
, call it type~$1$, and $(0, 1, 0)^T$%
, call it type~$2$.
There are three columns of type~$1$ and one column of type~$2$.
An optimal solution~$0001$ with minimum Mirkin distance four for $S$ can be encoded by two binary variables~$x[1]=0$ and $x[2]=1$.
\end{example}

The above considerations now yield the following.

\begin{theorem}
  \MirkMin{} can be solved in $O(2^{2^m} \cdot m \cdot n)$~time.
\end{theorem}
\begin{proof}
  The algorithm tries all at most $2^{2^m}$ possibilities for the binary string $x \in \{0, 1\}^{n'}$.
  For each of them, it constructs the corresponding solution string $\sol$ that in each column of type~$j \in [n']$ equals $x[j]$.
  This takes $O(n)$ time.
  The algorithm then computes the Hamming distance between $\sol$ and each of the input strings, which takes $O(n \cdot m)$ time in total.
  It then computes the Mirking distances from these Hamming distances in $O(m)$ time.
  Finally, it reports success if the sum of the Mirkin distances is at most $k$.
  It is clear that the running-time bound is satisfied.
  By \cref{lem:same-type-same-value} this algorithm will find a solution if there is one.
\end{proof}

\paragraph{Integer Linear Program Formulation}
Using the binary variables~$x$ that represent a solution~$\sol$ which has the same values in the columns of the same type, we can reformulate the Hamming distance between the two strings~$s_i$ and $\sol$ as follows.
For the sake of readability, we let \myemph{$s_i[j]=1$} if 
the column type of column~$j$ has $1$ in the $i^{\text{th}}$ row and \myemph{$s_i[j]=0$} if it has $0$ in the $i^{\text{th}}$ row.
\begin{align*}
  \hamming{}(s_i,\sol) & = \sum_{j=1}^{n'}e[j]\cdot |s[j]-x[j]| \\
   & = \sum_{j=1}^{n'}e[j] \cdot \left(s_i[j]+(1-2s_i[j])\cdot x[j]\right).
\end{align*}
Then, the Mirkin distance between $x$ and $s_i$ can be formulated as follows, 
where we let \myemph{$w_i$} $=\sum_{j=1}^{n'}e[j]\cdot s_i[j]$ denote the number of ones in string~$s_i$ and
\myemph{$c_i[j]$} $=1-2s_i[j]$, i.e., $c_i[j]=1$ if $s_i[j]=0$ and $c_i[j]=-1$ if $s_i[j]=1$.
\begin{align}
  &\mirkin(s_i,\sol) =  \hamming{}(s_i,\sol)\cdot (n-\hamming{}(s_i,\sol))\nonumber\\
  & =\left(w_i+ \sum_{j=1}^{n'} e[j] \cdot c_i[j]\cdot x[j]\right)\cdot  \left(n-w_i- \sum_{j=1}^{n'} e[j] \cdot c_i[j]\cdot x[j]\right)\nonumber\\
  & = n\left(w_i+\sum_{j=1}^{n'}e[j]\cdot c_i[j] \cdot x[j]\right)
    - \left(w_i+\sum_{j=1}^{n'}e[j]\cdot c_i[j] \cdot x[j]\right)^2\nonumber
  \\
                    & = n\cdot w_i - w_i^2 +\sum_{j=1}^{n'}\Biggl( n\cdot c_i[j]-2w_i\cdot c_i[j]-e[j]\Biggr)\cdot e[j]\cdot x[j] \nonumber\\
                    &
                      \quad - \sum_{{\substack{\{j,j'\}\subseteq [n']\\j\neq j'}}}e[j]\cdot e[j']\cdot c[j]\cdot c[j'] \cdot x[j] \cdot x[j']. 
\label{eq:mirkin}
\end{align}
The last equation holds since $c_i[j]^2 = 1$ and since $x[j]$ is binary, implying that $x[j]=x[j]^2$.

The resulting formulation is not linear since the components in the last sum are products of two binary variables.
Nevertheless, we can introduce additional binary variables to linearize it.
For each two distinct column types~$j$ and $j'$ we introduce a binary variable~$y[\{j,j'\}]$ which shall have the value $y[\{j,j'\}]=x[j]\cdot x[j']$.
We can achieve this by introducing the following constraints:
\begin{subequations}
  \begin{align}
  \forall j, j' \in [n'],  j\neq j' \colon & y[\{j,j'\}], x[j]\in \{0,1\},\label{eq:binary-y}\\
                                  & y[\{j,j'\}] \le x[j],\label{eq:product-1}\\
   & y[\{j,j'\}] \le x[j'],\label{eq:product-2}\\
 &  x[j] + x[j'] - y[\{j,j'\}] \le 1.\label{eq:product-3}
\end{align}
\end{subequations}
Now we can replace each product of two binary variables in \eqref{eq:mirkin} with a corresponding variable:
\begin{align}
  & \mirkin(s_i,\sol) \stackrel{\eqref{eq:mirkin}}{=} \nonumber \\
  & n\cdot w_i - w_i^2  +  \sum_{j=1}^{n'}\Biggl(n\cdot c_i[j]-2w_i\cdot c_i[j]-e[j]\Biggr)\cdot e[j]\cdot x[j]  \nonumber\\
  & - \sum_{\substack{\{j,j'\}\subseteq \{1,\ldots, n'\}\\j\neq j'}}e[j]\cdot e[j']\cdot c[j]\cdot c[j'] \cdot y[\{j,j'\}].
\end{align}
\noindent Combining \eqref{eq:binary-y}--\eqref{eq:product-3} with the following constraint

\setcounter{mycounter}{\theequation}
\setcounter{equation}{1}

\begin{subequations}
\begin{align}
\setcounter{equation}{4}
  \label{eq:objective-mirkin}
 & \sum_{i=1}^{m}\Biggl(n\cdot w_i - w_i^2  +  \sum_{j=1}^{n'}\biggl(n\cdot c_i[j]-2w_i\cdot c_i[j]-e[j]\biggr)\cdot e[j]\cdot x[j] \nonumber \\
 & - \sum_{\substack{\{j,j'\}\subseteq [n']\\j\neq j'}}e[j]\cdot e(j')\cdot c[j]\cdot c[j'] \cdot y[\{j,j'\}]\Biggr) \le k,
\end{align}
\end{subequations}
we obtain an ILP with at most $4^m + 2^m$ binary variables and $5n^2$ constraints, each with $O(m)$ terms.

\setcounter{equation}{\themycounter}

\section{Conclusion}
While we now know that \MirkMin\ can be solved in single-exponential time with respect to the length~$n$ of the input strings and such a running time is required, the basis of the exponential is not yet determined.
Is there an algorithm running in time $(2 - \epsilon)^n \cdot nm^{O(1)}$ or can we give a lower bound based on the Strong Exponential Time Hypothesis?

Currently, there is no known lower bound matching our fixed-parameter running time with respect to the number~$m$ of input strings.
Can \MirkMin\ be solved in $2^{m^{O(1)}} \cdot (mn)^{O(1)}$ time?

Finally, it is interesting to study other distance measures between the strings (or clusterings), such as the variation of information~\cite{meila_comparing_2007}, and to study other aggregation functions of the distance measures, such as taking a $p$-norm of the distance vector instead of the sum of distances~\cite{CheHerSor2018pnormesa}.

\section*{Acknowledgments}
\noindent 
The main work was done while both Jiehua Chen and Manuel Sorge were with Ben-Gurion University of the Negev, funded by the People Programme (Marie Curie Actions) of the European Union's Seventh Framework Programme (FP7/2007-2013)  under REA grant agreement number {631163.11} and Israel Science Foundation (grant number 551145/14). Jiehua Chen acknowledges support by the Vienna Science and Technology Fund (WWTF)~[10.47379/VRG18012].
Manuel Sorge acknowledges support by the Alexander von Humboldt Foundation. 

\bibliographystyle{abbrvnat}
\bibliography{bib}

\end{document}